\documentclass[11pt]{article}
\usepackage[english]{babel}
\usepackage{amsmath}
\usepackage{amsthm}
\usepackage{amssymb}
\usepackage[dvips]{color}
\newtheorem{theorem}{Theorem}[section]
\newtheorem{definition}[theorem]{Definition}
\newtheorem{problem}[theorem]{Problem}

\newtheorem{remark}[theorem]{Remark}
\newtheorem{example}[theorem]{Example}

\title{\bf A new method to generate superoscillating functions and supershifts}

\author{Y. Aharonov
\footnote{Schmid College of Science and Technology, Chapman University, Orange 92866, CA, US,
\ {\tt aharonov@chapman.edu, \ tollakse@chapman.edu}} ,\
 F. Colombo
 \footnote{Politecnico di Milano, Dipartimento di Matematica, Via E. Bonardi, 9 20133 Milano, Italy, \ {\tt fabrizio.colombo@polimi.it, \, irene.sabadini@polimi.it}} ,\
 I. Sabadini$^{\dagger}$ ,\
 T. Shushi\footnote{Department of Business Administration, Guilford Glazer Faculty of Business and Management, Ben-Gurion University of the Negev, Beer-Sheva, Israel, {\tt tomershu@bgu.ac.il}},\
  D.C. Struppa\footnote{The Donald Bren Distinguished Presidential Chair in Mathematics, Chapman University, Orange, USA, \ {\tt struppa@chapman.edu}} ,\
 J. Tollaksen$^*$
}

\begin{document}
\maketitle

\begin{abstract}
Superoscillations are  band-limited functions that  can oscillate faster than
their fastest Fourier component.
These functions (or sequences) appear in  weak values in quantum mechanics and in many fields of science and technology such as optics, signal processing and antenna theory.
 In this paper we introduce a new method to generate superoscillatory functions that allows us to construct explicitly a very large class of superoscillatory functions.
\end{abstract}
\vskip 1cm

\section{Introduction }
Superoscillating functions are  band-limited functions that  can oscillate faster than
their fastest Fourier component. Physical phenomena associated with superoscillatory functions
are known since long time and in more recent years there has been  a wide interest both from the physical and the mathematical point of view.
These functions (or sequences) appeared in  weak values in quantum mechanics, see \cite{aav,abook,b5},
in antenna theory this phenomenon was formulated  in  \cite{TDFG}.
The literature on superoscillations is large, and without claiming completeness we mention the papers
 \cite{AShushi},  \cite{berry2}-\cite{b4},  \cite{kempf1}-\cite{kempf2HHH} and \cite{lindberg}.
This class of functions has been investigated also from the mathematical point of view, as function theory, but large part of the results are associated with the study
of the evolution of superoscillations
 by quantum fields equations  with  particular attention to Schr\"odiger equation.
 We give a quite complete list of papers
 \cite{uno}, \cite{ABCS1}-\cite{acsst5},
 \cite{QS2}-\cite {Jussi}, \cite{cinque}-\cite{due}, \cite{sodakemp} where one can find an up-to-date panorama of this field.
  In order to have an overview of the techniques developed in the recent years to study the evolution of superoscillations
and their function theory, we refer the reader to the introductory papers \cite{QS1,QS3,QS2} and \cite{kempfQS}.
Finally we mention the {\em Roadmap on superoscillations}, see \cite{Be19}, where the most recent advances in superoscillations and their applications to technology are well explained by the leading experts in this field.

\medskip
A fundamental problem is to determine how large is the class of superoscillatory functions.
The prototypical superoscillating function that is the outcome of weak values is given by
\begin{equation}\label{FNEXP}
F_n(x,a)=\Big(\cos\Big(\frac{x}{n}\Big)+ia\sin \Big(\frac{x}{n}\Big) \Big)^n
=\sum_{j=0}^nC_j(n,a)e^{i(1-2j/n)x},\ \ x\in \mathbb{R},
\end{equation}
where $a>1$ and the coefficients $C_j(n,a)$ are given by
\begin{equation}\label{Ckna}
C_j(n,a)={n\choose j}\left(\frac{1+a}{2}\right)^{n-j}\left(\frac{1-a}{2}\right)^j.
\end{equation}
If we fix $x \in \mathbb{R}$  and we let $n$ go to infinity, we  obtain that
$$
\lim_{n \to \infty} F_n(x,a)=e^{iax}.
$$
Clearly the name superoscillations comes from the fact that
in the Fourier's representation of the function (\ref{FNEXP}) the frequencies
$1-2j/n$ are bounded by 1, but the limit function $e^{iax}$ has a frequency $a$ that can be arbitrarily larger  than $1$.
A precise definition of superoscillating functions is as follows.

\medskip
We call {\em generalized Fourier sequence}
a sequence of the form
\begin{equation}\label{basic_sequence}
f_n(x):= \sum_{j=0}^n X_j(n,a)e^{ih_j(n)x}, \ \ j=0,...,n, \ \ \ n\in \mathbb{N},
\end{equation}
where $a\in\mathbb R$, $X_j(n,a)$ and $h_j(n)$
are complex and real valued functions of the variables $n,a$ and $n$, respectively.
A generalized Fourier sequence of the form (\ref{basic_sequence})
 is said to be {\em a superoscillating sequence} if
 $\sup_{j,n}|k_j(n)|\leq 1$ and
 there exists a compact subset of $\mathbb R$,
 which will be called {\em a superoscillation set},
 on which $f_n(x)$ converges uniformly to $e^{ig(a)x}$,
 where $g$ is a continuous real valued function such that $|g(a)|>1$.

\medskip
The classical Fourier expansion is obviously not a superoscillating sequence since its frequencies are not, in general, bounded.
Using infinite order differential operators we can define
a class of superoscillatory function applying them to
the functions (\ref{FNEXP}) and we obtain superoscillating functions of the form
\begin{equation}\label{YN}
Y_n(x,a)=\sum_{j=0}^nC_j(n,a)e^{ig(1-2j/n)x},
\end{equation}
where $C_j(n,a)$ are the coefficients in (\ref{Ckna}),  $g$ are given entire functions, monotone increasing in $a$, and $x\in \mathbb{R}$.
We have shown that
$$
\lim_{n\to\infty}Y_n(x,a)=e^{ig(a)x}
$$
under suitable conditions on $g$
and the simplest, but important,  example is
$$
Y_n(x,a)=\sum_{j=0}^nC_j(n,a)e^{i(1-2j/n)^mx}, \ \ {\rm for\ fixed} \ \ m\in \mathbb{N}.
$$
From the above considerations we deduce that there exists a large class of
superoscillating functions taking different functions $g$.

\medskip
In this paper we further enlarge the class of superoscillating functions enlarging both the
class of the coefficients  $C_j(n,a)$ and of the sequence of frequencies $1-2j/n$ that are bounded by 1.
A large class of superoscillating functions can be determined solving the following problem.
\begin{problem}\label{gsgdfaA}
Let $h_j(n)$ be a given set of points in $[-1,1]$, $j=0,1,...,n$, for $n\in \mathbb{N}$ and let $a\in \mathbb{R}$ be such that $|a|>1$.
Determine the coefficients $X_j(n)$ of the sequence
$$
f_n(x)=\sum_{j=0}^nX_j(n)e^{ih_j(n)x},\ \ \ x\in \mathbb{R}
$$
in such a way that
$$
f_n^{(p)}(0)=(ia)^p,\ \ \ {\rm for} \ \ \ p=0,1,...,n.
$$
\end{problem}
\begin{remark}
The conditions $f_n^{(p)}(0)=(ia)^p$  mean that the functions $x\mapsto e^{iax}$ and $f_n(x)$ have the same derivatives at the origin, for
$p=0,1,...,n$,  so they have the same Taylor polynomial of order $n$.

\end{remark}

Under the condition that the points  $h_j(n)$  for $j=0,...,n$, (often denoted by $h_j$) are distinct we
obtain an explicit formula  for the coefficients  $X_j(n,a)$ given by
$$
X_j(n,a)=
\prod_{k=0,\  k\not=j}^n\Big(\frac{h_k(n)-a}{h_k(n)-h_j(n)}\Big),
$$
so the superoscillating sequence $f_n(x)$, that solves Problem \ref{gsgdfaA}, takes  the explicit  form
$$
f_n(x)=\sum_{j=0}^n\prod_{k=0,\  k\not=j}^n\Big(\frac{h_k(n)-a}{h_k(n)-h_j(n)}\Big)\ e^{ih_j(n)x},\ \ \ x\in \mathbb{R},
$$
as shown in Theorem \ref{exsolution}.
Observe that, by construction, this function is band limited and it converges to $e^{iax}$ with arbitrary $|a|>1$, so it is superoscillating.

\medskip
Observe that different sequences $X_j(n)$ can be explicitly computed when we fix the points $h_j(n)$.
See, for example, the case of the sequence
$$
 h_j(n)=1-2j/n^p \   {\rm for} \ j=0,...,n, \ \ \ n\in \mathbb{N} \ \ {\rm and}\  {\rm for\ fixed}\
  \ p\in \mathbb{N},
$$
in Section \ref{EXSUPER}.

\medskip
We consider now the frequencies
$h_j(n)=(1-2j/n)^m$, for fixed $m\in \mathbb{N}$,
to explain some facts.

(I) If we consider the sequence (\ref{YN}), with coefficients $C_j(n,a)$ given by (\ref{Ckna}), we obtain
$$
\lim_{n\to\infty}\sum_{j=0}^nC_j(n,a)e^{i(1-2j/n)^mx}=e^{ia^mx},    \ \ \ \  {\rm for\ fixed} \ \ m\in \mathbb{N}.
$$
 Note that, in this case, we could have used the frequencies $h_j(n)=(1-2j/n)$  and the coefficients $\tilde C_j(n,a):=C_j(n,a^m)$ to get as limit function $e^{ia^mx}$.
 Thus the same limit function  $e^{ia^mx}$ can be obtained by tuning the frequencies and the coefficients.

(II)
By  solving Problem \ref{gsgdfaA} with the frequences
$h_j(n)=(1-2j/n)^m$, we can determine the coefficients $X_j(n)=X_j(n,a)$ such that we obtain as limit function $e^{iax}$, namely
\begin{equation}\label{sdkgfo}
\lim_{n\to\infty}\sum_{j=0}^n{X}_j(n,a)e^{i(1-2j/n)^mx}=e^{iax}, \ \ \ {\rm for\ fixed } \ \  m\in \mathbb{N}.
\end{equation}
Changing the coefficients ${\tilde{X}}_j(n,a)$ we can get, as limit function, $e^{ia^mx}$.

(III) The coefficients $X_j$ and $C_j$ in the procedures (I) and (II) are different form each other
 because the two methods to generate superoscillations are different  as explained in Section 3.

\medskip
In section 4 we will also discuss how to generalize this method to obtain analogous results in the case of the supershift property of functions, a mathematical concept that generalizes the notion of superoscillating function.

\section{A new class of superoscillating functions}

In this section we show the main procedure to determine the coefficients $X_j(n)$ and so to construct explicitly the superoscillating functions
solving Problem \ref{gsgdfaA}.

\begin{theorem}[Existence and uniqueness of the solution of Problem \ref{gsgdfaA}]
\label{gsgdfathm}
Let $h_j(n)$ be a given set of points in $[-1,1]$, $j=0,1,...,n$ for $n\in \mathbb{N}$ and let $a\in \mathbb{R}$ be such that $|a|>1$.
If $h_j(n)\not= h_i(n)$, for every $i\not=j$,
then there exists a unique solution $X_j(n)$ of the linear system
$$
f_n^{(p)}(0)=(ia)^p,\ \ \ {\rm for} \ \ \ p=0,1,...,n,
$$
 in Problem \ref{gsgdfaA}.
\end{theorem}
\begin{proof}
For the sake of simplicity, we denote $h_j(n)$ by $h_j$.
Observe that the derivatives of order $p$ of  $f_n(x)$ are
$$
f^{(p)}_n(x)=\sum_{j=0}^nX_j(n)(ih_j)^pe^{ih_jx},\ \ \ x\in \mathbb{R},
$$
so if we require that these derivatives are equal to the
derivatives of order $p$ for $p=0,1,...,n$ of
the function $x\mapsto e^{iax}$ at the origin we obtain the linear system
\begin{equation}\label{linsyst}
\sum_{j=0}^nX_j(n)(ih_j)^p=(ia)^p,\ \ \ \ p=0,1,...,n
\end{equation}
from which we deduce
\begin{equation}\label{sistlin}
\sum_{j=0}^nX_j\ (h_j)^p=a^p,\ \ \ \ p=0,1,...,n,
\end{equation}
where we have written $X_j$ instead of $X_j(n)$.
Now we write  explicitly the linear system (\ref{sistlin}) of $(n+1)$ equations and $(n+1)$ unknowns
$(X_0,...,X_n)$
\[
\begin{split}
X_0+X_1+ \ldots +X_n&=1
\\
X_0h_0+X_1h_1+ \ldots +X_nh_n&=a
\\
&
\ldots
\\
X_0h_0^n+X_1h_1^n+ \ldots +X_nh_n^n&=a^n
\end{split}
\]
and,
 in matrix form, it becomes
\begin{equation}\label{LSYT}
H(n)X=B(a)
\end{equation}
where $H$ is the
$(n+1)\times (n+1)$ matrix
\begin{equation}
H(n)=\begin{pmatrix} 1 & 1 & \ldots & 1
\\
h_0& h_1&  \ldots & h_n
 \\
\ldots & \ldots&  \ldots & \ldots
\\
h_0^n& h_1^n&  \ldots & h_n^n
 \end{pmatrix}
\end{equation}
and
\begin{equation}
X=\begin{pmatrix} X_0
\\
X_1
 \\
\ldots
\\
X_n
 \end{pmatrix}
 \ \ \
 {\rm and}
 \ \ \
 B(a)=\begin{pmatrix} 1
\\
a
 \\
\ldots
\\
a^n
 \end{pmatrix}.
\end{equation}
Observe that the determinant of $H$ is the Vandermonde determinant, so it is given by
$$
\det(H(n))=\prod_{0\leq i< j\leq n}(h_j(n)-h_i(n)).
$$
Thus, if $h_j(n)\not= h_i(n)$ for every $i\not=j$, there exists a unique solution of the system. \end{proof}

\begin{theorem}[Explicit solution of Problem \ref{gsgdfaA}]\label{exsolution}
Let $h_j(n)$ be a given set of points in $[-1,1]$, $j=0,1,...,n$ for $n\in \mathbb{N}$ and let $a\in \mathbb{R}$ be such that $|a|>1$.
If $h_j(n)\not= h_i(n)$, for every $i\not=j$,
the unique solution of system (\ref{LSYT})
is given by
\begin{equation}\label{sxplsolut}
X_j(n,a)=
\prod_{k=0,\  k\not=j}^n\Big(\frac{h_k(n)-a}{h_k(n)-h_j(n)}\Big).
\end{equation}
As a consequence, the superoscillating function takes the form
$$
f_n(x)=\sum_{j=0}^n\prod_{k=0,\  k\not=j}^n\Big(\frac{h_k(n)-a}{h_k(n)-h_j(n)}\Big)\ e^{ih_jx},\ \ \ x\in \mathbb{R}.
$$
\end{theorem}
\begin{proof}
In Theorem \ref{gsgdfathm} we proved that, if $h_j\not= h_i$ for every $i\not=j$, there exists a unique solution of the system (\ref{LSYT}). The solution is given by
\begin{equation}\label{gghhhh}
X_j(n,a)=\frac{\det(H_j(n,a))}{\det(H(n))}
\end{equation}
for
\begin{equation}
H_j(n,a)=\begin{pmatrix} 1 & 1 & \ldots & 1 & \ldots &  1
\\
h_0& h_1&   \ldots & a & \ldots &  h_n
 \\
\ldots & \ldots &  \ldots & \ldots& \ldots
\\
h_0^n& h_1^n&  \ldots & a^n & \ldots &  h_n^n
 \end{pmatrix}
\end{equation}
where the $j^{th}$-column contains $a$ and its powers.
The explicit form of the determinant of the matrix $H$ is given by:
$$
det(H(n))=(h_1-h_0)\cdot (h_2-h_0)(h_2-h_1) \cdot
(h_3-h_0)(h_3-h_1)(h_3-h_2)
$$
$$\cdot (h_4-h_0)(h_4-h_1)(h_4-h_2)(h_4-h_3)\cdot \ldots
\cdot (h_n-h_0)(h_n-h_1)(h_n-h_2)(h_n-h_3) ..... (h_n-h_{n-1}).
$$

The matrix $H_j(n,a)$ is still of Vandermonde type and its determinant can be computed similarly.
 So we have that the solution $(X_0(n,a),\ldots, X_n(n,a))$ is such that
\[
\begin{split}
X_0(n,a)&=\frac{(h_1-a)\cdot (h_2-a) \cdot(h_3-a)\cdot (h_4-a)\cdot \ldots\cdot (h_n-a)}{
(h_1-h_0)\cdot (h_2-h_0) \cdot(h_3-h_0)\cdot (h_4-h_0)\cdot \ldots\cdot (h_n-h_0)}
\\
&
=\frac{\prod_{k=0,\  k\not=0}^n(h_k-a)}{\prod_{k=0,\  k\not=0}^n(h_k-h_0)},
\end{split}
\]
\[
\begin{split}
X_1(n,a)&=
\frac{(a-h_0)\cdot (h_2-a) \cdot(h_3-a)\cdot (h_4-a)\cdot \ldots\cdot (h_n-a)
}{
(h_1-h_0)\cdot (h_2-h_1) \cdot(h_3-h_1)\cdot (h_4-h_1)\cdot \ldots\cdot (h_n-h_1)}
\\
&
=\frac{\prod_{k=0,\  k\not=1}^n(h_k-a)}{\prod_{k=0,\  k\not=1}^n(h_k-h_1)},
\end{split}
\]
and so on, up to
\[
\begin{split}
X_n(n,a)&=
\frac{1\cdot 1 \cdot1\cdot 1\cdot \ldots\cdot (a-h_0)(a-h_1)(a-h_2)(a-h_3) ..... (a-h_{n-1})
}{1\cdot 1 \cdot1\cdot 1\cdot \ldots\cdot (h_n-h_0)(h_n-h_1)(h_n-h_2)(h_n-h_3) ..... (h_n-h_{n-1})}
\\
&
=\frac{\prod_{k=0,\  k\not=n}^n(h_k-a)}{\prod_{k=0,\  k\not=n}^n(h_k-h_n)}.
\end{split}
\]
So we get the statement with the recursive formula.
\end{proof}

\section{The methods to generate superoscillations and examples}\label{EXSUPER}

Below we compare the superoscillating functions obtained by solving Problem
\ref{gsgdfaA} and
the superoscillating functions obtained via the sequence $F_n(x,a)$ and infinite order differential operators.
For different methods see also \cite{kempf2HHH}.

\medskip
(I)
Observe that the limit
$$
\lim_{n\to\infty}\Big(\cos\Big(\frac{x}{n}\Big)+ia\sin \Big(\frac{x}{n}\Big) \Big)^n=e^{iax}
$$
is a direct consequence of
$$
\lim_{n\to\infty}\Big(1+ia\frac{x}{n}\Big)^{n}=e^{iax},
$$
while the construction method to generate superoscillations in Theorem \ref{exsolution} has a different nature
 because we require that the linear system (\ref{linsyst}) in the $n+1$ unknowns $X_j(n)$ are determined in such a way that
\begin{equation}
\sum_{j=0}^nX_j(n)(ih_j)^p=(ia)^p,\ \ \ \ p=0,1,...,n
\end{equation}
so the derivatives
$$
f^{(p)}_n(x)=\sum_{j=0}^nX_j(n)(ih_j)^pe^{ih_j(n)x},\ \ \ x\in \mathbb{R},
$$
at $x=0$, are equal to the derivatives of the exponential function $e^{iax}$ at the origin.
This means that the sequence of functions
$$
f_n(x)=\sum_{j=0}^nX_j(n)e^{ih_j(n)x},\ \ \ x\in \mathbb{R}
$$
has $n$ derivatives equal to the derivatives of exponential function  $e^{iax}$ at the origin
so the limit
$$
\lim_{n\to\infty}f_n(x)=e^{iax}
$$
 follows by construction of the $f_n(x)$.

\medskip
(II) In the definition of the superoscillating function
(\ref{FNEXP}) the derivatives are given by
$$
F_n^{(p)}(x,a)
=\sum_{j=0}^nC_j(n,a)\Big(i(1-2j/n)\Big)^pe^{i(1-2j/n)x},\ \ x\in \mathbb{R}
$$
and it is just in the limit that we get the derivatives of order $p$ of the exponential function $e^{iax}$
at the origin, namely we have
$$
\lim_{n\to \infty}\sum_{j=0}^nC_j(n,a)\Big(i(1-2j/n)\Big)^p=(ia)^p, \ \ \ p\in \mathbb{N}.
$$

\medskip
(III)
 With the new procedure proposed in this paper  we impose the conditions
$$
f^{(p)}_n(0)=(ia)^p, \ \ \ p=0,1,2,...,n
$$
(where we have genuine equalities, not in the limit)
and we link $n$ with the derivatives of $f_n(x)$ in order to determine the coefficients $X_j(n)$ in
(\ref{basic_sequence}), so
we have that the Taylor polynomials of the two functions $f_n(x)$ and $e^{iax}$
are the same up to order $n$, i.e.,
$$
e^{iax}=1+iax+\frac{(iax)^2}{2!}+...+\frac{(iax)^n}{n!}+R_{n}(x),
$$
so we get
\[
\begin{split}
f_n(x)&=\sum_{j=0}^nX_j(n)e^{ih_jx}
\\
&
= f_n(0)+f_n^{(1)}(0)x+f_n^{(2)}(0)\frac{x^2}{2!}+...+
 f_n^{(n)}(0)\frac{x^n}{n!}+\tilde{R}_{n}(x),\ \ \ x\in \mathbb{R}
\\
&
= 1+ia x+(ia)^2\frac{x^2}{2!}+...+
 (ia)^n\frac{x^n}{n!}+\tilde{R}_{n}(x),\ \ \ x\in \mathbb{R},
\end{split}
\]
where ${R}_{n}(x)$ and
$\tilde{R}_{n}(x)$ are the errors.

\medskip

It is now easy to generate a very large class of superoscillatory functions.
We write a few examples to further clarify the generality of our new construction
of superoscillating sequences: given the sequence $h_j(n)$ we determine the coefficients accordingly
for the
$$
f_n(x)=\sum_{j=0}^n\prod_{k=0,\  k\not=j}^n\Big(\frac{h_k(n)-a}{h_k(n)-h_j(n)}\Big)\ e^{ih_j(n)x},\ \ \ x\in \mathbb{R}.
$$

\begin{example}
Let $n\in\mathbb N$ and set
$$
h_j(n)=1-\frac{2}{n}j$$ where $j=0,...,n$. We have
$$
h_k(n)-a=1-\frac{2}{n}k -a
$$
and
$$
h_k(n)-h_j(n)=1-\frac{2}{n}k -\Big( 1-\frac{2}{n}j\Big)=\frac{2}{n}\Big(j-k\Big).
$$
Thus, we obtain
$$
f_n(x)=\sum_{j=0}^n\prod_{k=0,\  k\not=j}^n\frac{n}{2}\Big(\frac{1-\frac{2}{n}k-a}{j-k}\Big)\ e^{i(1-\frac{2}{n}j)x},\ \ \ x\in \mathbb{R}.
$$
\end{example}

\begin{example}
Let $n\in\mathbb N$, and set  $$h_j(n)=1-\frac{2}{n^p}j$$ where $j=0,...,n$, for a fixed $p\in \mathbb{N}$. We have
$$
h_k(n)-a=1-\frac{2}{n^p}k -a,
$$
and
$$
h_k(n)-h_j(n)=1-\frac{2}{n^p}k -\Big( 1-\frac{2}{n^p}j\Big)=\frac{2}{n^p}\Big(j-k\Big).
$$
So, we obtain:
$$
f_n(x)=\sum_{j=0}^n\prod_{k=0,\  k\not=j}^n\frac{n^p}{2}\Big(\frac{1-\frac{2}{n^p}k-a}{j-k}\Big)\ e^{i(1-\frac{2}{n^p}j)x},\ \ \ x\in \mathbb{R}.
$$
\end{example}

\begin{example}
Let $n\in\mathbb N$, and set  $$h_j(n)=1-\Big(\frac{2j}{n}\Big)^p$$ where $j=0,...,n,$ for a fixed $p\in \mathbb{N}$. We have
$$
h_k(n)-a=1-\Big(\frac{2k}{n}\Big)^p -a,
$$
and
$$
h_k(n)-h_j(n)=1-\Big(\frac{2k}{n}\Big)^p -\Big( 1-\Big(\frac{2j}{n}\Big)^p\Big)=
\frac{2^p}{n^p}\Big(j^p-k^p\Big).
$$
So, we obtain
$$
f_n(x)=\sum_{j=0}^n\prod_{k=0,\  k\not=j}^n\frac{n^p}{2^p}\Big(\frac{1-\frac{2k}{n}-a}{j^p-k^p}\Big)\ e^{i(1-(2j/n)^p)x},\ \ \ x\in \mathbb{R}.
$$
\end{example}

\section{ A new class of supershifts}

 The procedure to define superoscillatory functions can be extended to supershift.
 We recall that the supershift property of a function extends the notion
 of superoscillations and it turned out to be the crucial concept behind the study of the evolution
 of superoscillatory functions as initial conditions of Schr\"odinger equation or of any other field equation.
We recall the definition before to state our result.
\begin{definition}[Supershift]\label{Super-shift}
Let $\mathcal{I}\subseteq\mathbb{R}$ be an interval with $[-1,1]\subset \mathcal I$
and let
$\varphi:\, \mathcal I  \times \mathbb{R}\to \mathbb R$, be a continuous function on $\mathcal I$.
We set
$$
\varphi_{\lambda}(x):=\varphi(\lambda,x), \ \ \lambda \in \mathcal{I},\ \ x\in \mathbb R
 $$
and we consider a sequence of points $(\lambda_{j,n})$ such that
 $$
  (\lambda_{j,n})\in [-1,1] \ \  {\rm for} \ \  j=0,...,n \ \ {\rm  and} \ \ n=0,\ldots,+\infty.
 $$
   We define the functions
\begin{equation}\label{psisuprform}
\psi_n(x)=\sum_{j=0}^nc_j(n)\varphi_{\lambda_{j,n}}(x),
\end{equation}
where  $(c_j(n))$ is a sequence of complex numbers for $j=0,...,n$ and $n=0,\ldots,+\infty$.
If
$
\lim_{n\to\infty}\psi_n(x)=\varphi_{a}(x)
$
for some $a\in\mathcal I$ with $|a|>1$, we say that the function
$\lambda\to \varphi_{\lambda}(x)$, for  $x$ fixed, admits a supershift in $\lambda$.
 \end{definition}

 \begin{remark}
 We observe that the definition of supershift of a function given above is not the most general one, but it is useful to explain our new procedure for the supershift case.
 In the following,  we will take the interval $\mathcal{I}$, in the definition of the supershift, to be equal to $\mathbb{R}$.
 \end{remark}

 \begin{remark}
 Let us stress that the term supershift comes from the fact that the interval
  $\mathcal I$ can be arbitrarily large (it can also be $\mathbb R$)
  and so also the constant $a$ can be arbitrarily far away from the interval $[-1,1]$ where the function $\varphi_{\lambda_{j,n}}(\cdot)$ is computed, see \eqref{psisuprform}.
 \end{remark}

 \begin{remark}
 Superoscillations are a particular case of supershift. In fact,
   for the sequence $(F_n)$ in \eqref{FNEXP}, we have $\lambda_{j,n}=1-2j/n$, $\varphi_{\lambda_{j,n}}(t,x)=e^{i(1-2j/n)x}$ and  $c_j(n)$ are the coefficients
$C_j(n,a)$ defined in (\ref{Ckna}).
 \end{remark}

Problem \ref{gsgdfaA}, for the supershift case, is formulated as follows.

\begin{problem}\label{gsgdfaSUP}
Let $h_j(n)$ be a given set of points in $[-1,1]$, $j=0,1,...,n$, for $n\in \mathbb{N}$ and let $a\in \mathbb{R}$ be such that $|a|>1$.
Suppose that for every $x\in \mathbb{R}$ the function $\lambda\mapsto G(\lambda x )$
is holomorphic and entire in $\lambda$.
Consider  the function
$$
f_n(x)=\sum_{j=0}^nY_j(n)G(h_j(n)x),\ \ \ x\in \mathbb{R}
$$
where
$\lambda\mapsto G(\lambda x)$ depends on the parameter $x\in \mathbb{R}$.
Determine the coefficients $Y_j(n)$
in such a way that
\begin{equation}\label{eqsuper}
f_n^{(p)}(0)=(a)^pG^{(p)}(0)\ \ \ for \ \ \ p=0,1,...,n.
\end{equation}
\end{problem}
The solution of the above problem is obtained in the following theorem.
\begin{theorem}
Let $h_j(n)$ be a given set of points in $[-1,1]$, $j=0,1,...,n$ for $n\in \mathbb{N}$ and let $a\in \mathbb{R}$ be such that $|a|>1$.
If $h_j(n)\not= h_i(n)$ for every $i\not=j$ and $G^{(p)}(0)\not=0$ for all $p=0,1,...,n$,
then there exists a unique solution $Y_j(n,a)$ of the linear system (\ref{eqsuper})
and it is given by
$$
Y_j(n,a)=
\prod_{k=0,\  k\not=j}^n\Big(\frac{h_k(n)-a}{h_k(n)-h_j(n)}\Big),
$$
so that
$$
f_n(x)=\sum_{j=0}^n\prod_{k=0,\  k\not=j}^n\Big(\frac{h_k(n)-a}{h_k(n)-h_j(n)}\Big)G(h_j(n)x),\ \ \ x\in \mathbb{R}
$$
and, by construction, it is
$$
\lim_{n\to\infty}f_n(x)=G(ax), \ \ x\in\mathbb R.
$$
\end{theorem}
\begin{proof}
Observe that we have
$$
f_n^{(p)}(x)=\sum_{j=0}^nY_j(n)(h_j(n))^p G^{(p)}(h_j(n) x),\ \ \ x\in \mathbb{R}
$$
where $G^{(p)}$ are the derivatives  of order $p$, for $p=0,...,n$  with respect to $x$ of the function
$G(\lambda x)$, for $\lambda\in \mathbb{R}$ considered as a parameter.
So we get the system
$$
f_n^{(p)}(0)=\sum_{j=0}^nY_j(n)(h_j(n))^p G^{(p)}(0)=a^pG^{(p)}(0).
$$
Now, since we have assumed that $G^{(p)}(0)\not=0$ for all $p=0,1,...,n$, the system becomes
$$
\sum_{j=0}^nY_j(n)(h_j(n))^p =a^p
$$
and we can solve it as in Theorem \ref{exsolution} to get
$$
Y_j(n,a)=
\prod_{k=0,\  k\not=j}^n\Big(\frac{h_k(n)-a}{h_k(n)-h_j(n)}\Big).
$$
Finally, we get
$$
f_n(x)=\sum_{j=0}^n\prod_{k=0,\  k\not=j}^n\Big(\frac{h_k(n)-a}{h_k(n)-h_j(n)}\Big)G(h_j(n)x),\ \ \ x\in \mathbb{R}
$$
and by construction it is
$$
\lim_{n\to\infty}f_n(x)=G(ax), \ \ x\in\mathbb{R}.
$$
\end{proof}

\end{document}